\documentclass[conference,letterpaper]{IEEEtran}

\usepackage{bm}
\usepackage{amsmath,amssymb,amsfonts,graphicx,amsthm,mathtools,nicefrac,mathrsfs}
\usepackage{graphicx, subfigure}
\usepackage{times}
\usepackage{color}
\usepackage[table]{xcolor}
\usepackage{booktabs}
\usepackage{multirow}
\usepackage{lscape}

\usepackage[utf8]{inputenc}
\usepackage[T1]{fontenc}
\usepackage{url}
\usepackage{ifthen}
\usepackage{cite}
\DeclareMathOperator{\E}{\mathbb{E}}

\DeclareMathOperator{\rad}{rad}
\DeclareMathOperator{\diam}{diam}

\renewcommand{\Pr}[2][]{\mathbb{P}_{#1} \left\{ #2 \rule{0mm}{3mm}\right\}}
\newcommand{\ip}[2]{\left\langle#1,#2\right\rangle}

\def \DD {\mathcal{D}}
\def \EE {\mathcal{E}}
\def \NN {\mathcal{N}}

\def \TT {\mathcal{T}}
\def \SS {\mathcal{S}}
\def \KK {\mathcal{K}}

\def \B {\mathbb{B}}
\def \R {\mathbb{R}}
\def \S {\mathbb{S}}

\def \va {\bm{a}}
\def \vb {\bm{b}}

\def \ve {\bm{e}}
\def \vg {\bm{g}}
\def \vh {\bm{h}}
\def \vn {\bm{n}}

\def \vu {\bm{u}}
\def \vv {\bm{v}}
\def \vw {\bm{w}}
\def \vx {\bm{x}}
\def \vy {\bm{y}}
\def \vz {\bm{z}}

\def \vzero {\bm{0}}

\def \mA {\bm{A}}

\def \mI {\bm{I}}

\def \mP {\bm{P}}

\def \mPhi {\bm{\Phi}}

\newtheorem{theorem}{Theorem}
\newtheorem{lemma}{Lemma}

\newtheorem{fact}{Fact}

\theoremstyle{definition}
\newtheorem{definition}{Definition}

\theoremstyle{remark}
\newtheorem{remark}{Remark}

\makeatletter
\newtheorem*{rep@theorem}{\rep@title}
\newcommand{\newreptheorem}[2]{%
	\newenvironment{rep#1}[1]{%
		\def\rep@title{#2 \ref{##1}}%
		\begin{rep@theorem}}%
		{\end{rep@theorem}}}
\makeatother

\newreptheorem{theorem}{Theorem}
\newreptheorem{lemma}{Lemma}
\newreptheorem{corollary}{Corollary}
\newreptheorem{proposition}{Proposition}
\begin{document}
%\title{Non-Linear Corrupted Sensing with Geometric Constraints}
\title{Recovery of Structured Signals From Corrupted Non-Linear Measurements}
% %%% Single author, or several authors with same affiliation:
% \author{
%   \IEEEauthorblockN{Stefan M.~Moser}
%   \IEEEauthorblockA{ETH Zürich\\
%                    ISI (D-ITET)\\
 %                   CH-8092 Zürich, Switzerland\\
%                    Email: moser@isi.ee.ethz.ch}
%}
%%% Several authors with up to three affiliations:
\author{
  \IEEEauthorblockN{Zhongxing~Sun and Wei~Cui}
  \IEEEauthorblockA{School of Information and Electronics\\
  	                Beijing Institute of Technology\\
  	                Beijing 100081, China\\
                    Email: \{zhongxingsun,~cuiwei\}@bit.edu.cn}
  \and
  \IEEEauthorblockN{Yulong~Liu}
  \IEEEauthorblockA{School of Physics\\
                    Beijing Institute of Technology\\
                    Beijing 100081, China\\
                    Email: yulongliu@bit.edu.cn}
}
%%% Many authors with many affiliations:
% \author{%
%   \IEEEauthorblockN{Albus Dumbledore\IEEEauthorrefmark{1},
%                     Olympe Maxime\IEEEauthorrefmark{2},
%                     Stefan M.~Moser\IEEEauthorrefmark{3}\IEEEauthorrefmark{4},
%                     and Harry Potter\IEEEauthorrefmark{1}}
%   \IEEEauthorblockA{\IEEEauthorrefmark{1}%
%                     Hogwarts School of Witchcraft and Wizardry,
%                     1714 Hogsmeade, Scotland,
%                     \{dumbledore, potter\}@hogwarts.edu}
%   \IEEEauthorblockA{\IEEEauthorrefmark{2}%
%               r      Beauxbatons Academy of Magic,
%                     1290 Pyrénées, France,
%                     maxime@beauxbatons.edu}
%   \IEEEauthorblockA{\IEEEauthorrefmark{3}%
%                     ETH Zürich, ISI (D-ITET), ETH Zentrum,
%                     CH-8092 Zürich, Switzerland,
%                     moser@isi.ee.ethz.ch}
%   \IEEEauthorblockA{\IEEEauthorrefmark{4}%
%                     National Chiao Tung University (NCTU),
%                     Hsinchu, Taiwan,
%                     moser@isi.ee.ethz.ch}
% }
\maketitle
%%%%%%
%% Abstract:
%% If your paper is eligible for the student paper award, please add
%% the comment "THIS PAPER IS ELIGIBLE FOR THE STUDENT PAPER
%% AWARD." as a first line in the abstract.
%% For the final version of the accepted paper, please do not forget
%% to remove this comment!
%%
\begin{abstract}
This paper studies the problem of recovering a structured signal from a relatively small number of corrupted non-linear measurements. Assuming that signal and corruption are contained in some structure-promoted set, we suggest an extended Lasso to disentangle signal and corruption. We also provide conditions under which this recovery procedure can successfully reconstruct both signal and corruption.
\end{abstract}

\section{Introduction}

Throughout science and engineering, one is often faced with the challenge of recovering a structured signal from a relatively small number of \emph{linear} observations
\begin{equation*}
  \vy =  \mPhi \vx^{\star} + \vn,
\end{equation*}
where $\mPhi \in \R^{m \times n}$ is the sensing matrix, $\vx^{\star} \in \R^{n}$ is the desired structured signal, and $\vn \in \R^{m}$ is the random noise. The objective is to estimate $\vx^{\star}$ from given knowledge of $\vy$ and $\mPhi$. Since this problem is generally ill-posed, tractable recovery is possible when the signal is suitably structured. A general model to encode signal structure is to assume that $\vx^{\star}$ belongs to some set $\SS\subset \R^n$. For example, to promote sparsity (or low-rankness) of the solution, one can choose $\SS$ to be a scaled $\ell_1$ (or nuclear norm) ball. Then the signal can be recovered by solving the following $\SS$-Lasso problem:
\begin{align}\label{Lasso}
\min_{\vx} ~\|\vy - \mPhi\vx \|_2,\quad \text{s.t.~}& \vx \in \SS.
\end{align}
The performance of $\SS$-Lasso (and its variants) under linear measurements has been extensively studied in the literature, see e.g., \cite{chandrasekaran2012convex, tropp2015convex, vershynin2015estimation, thrampoulidis2015recovering} and references therein.

However, in many applications of interest the linear model may not be plausible. Important examples include $1$-bit compressed sensing \cite{boufounos20081} and generalized linear models \cite{mccullagh1984generalized}. In these scenarios, measurements can be approached with the semiparametric single index model \cite{ichimura1993semiparametric, horowitz1996direct}
\begin{equation}\label{eq: model1}
\vy_i = f_i(\ip{\mPhi_i}{\vx^{\star}}), ~~~~i  =  1, \ldots, m,
\end{equation}
where $f_i: \R \rightarrow \R$ are independent copies of an unknown \emph{non-linear} map $f$ (or it may be deterministic) and $\mPhi_i^T$ denote rows of $\mPhi$. In a seminal paper \cite{Plan2015The}, Plan and Vershynin present a theoretical analysis for $\SS$-Lasso under the non-linear observation model \eqref{eq: model1}. Their results show that non-linear observations behave as scaled and noisy linear observations, and under suitable conditions, a scaled original signal can be recovered by $\SS$-Lasso.

This work extends that of \cite{Plan2015The} to a more challenging setting, in which the non-linear measurements are corrupted by an unknown but structured vector $\vv^{\star}$, i.e.,
\begin{equation}\label{eq: model}
\vy_i = f_i(\ip{\mPhi_i}{\vx^{\star}})+\sqrt{m}\vv_i^{\star}, ~~~~i  =  1, \ldots, m.
\end{equation}
This model is motivated by some practical applications:
\begin{itemize}
	\item Clipping or saturation noise: signal clipping or saturation frequently appears in power-amplifiers and analog-to-digital converters (ADC) because of the limited range in the devices \cite{Abel1991Restoring, Laska2011Democracy}. In those cases, one always measures $f(\mPhi\vx)$ rather than $\mPhi\vx$, where $f$ is typically a nonlinear map. And saturation occurs when the input exceeding the maximum or minimum device output. Unlike the white noise or quantization error, the saturation can be unbounded. However, it will be sparse provided the clipping
	level is high enough, which means the model \eqref{eq: model} is appropriate. The elimination of saturation effect may be difficult in a broad class of radar and sonar systems \cite{Gray1971Quantization}.
	\item State estimation for electrical power networks: non-linear measurements $f(\vx)$ caused by device constraints are sent to
	the central control center in powers networks. These measurements may contain gross errors or outliers modeled by structured corruptions which have arbitrary amplitude due to system malfunctions. So state estimation in power networks needs to detect and eliminate these large measurement errors \cite{Handschin1974Bad, Xuweiyu2013Sparse, Broussolle2007State, Monticelli1983Reliable}.
\end{itemize}
In particular, if $f$ is the identity function, the model \eqref{eq: model} reduces to the standard corrupted sensing problem \cite{Foygel2014Corrupted, Mccoy2014Sharp, Amelunxen2014Living, Chen2017Corrupted, Zhang2017On, Jinchi2018Stable}.

Assume that $(\vx^{\star}, \vv^{\star})$ belongs to some set $\TT\subset\R^n\times\R^m$ which is meant to capture structures of signal and corruption. A natural method to disentangle signal and corruption is to minimize the $\ell_2$ loss subject to a geometric constraint:
\begin{align}\label{T_Lasso}
\min_{\vx, \vv} ~\|\vy-\mPhi\vx-\sqrt{m}\vv\|_2,\quad\text{s.t.~}& (\vx,\vv) \in \TT.
\end{align}
This procedure might be regarded as an extension of $\SS$-Lasso \cite{Plan2015The}.

The goal of this paper is to investigate the performance of $\TT$-Lasso \eqref{T_Lasso} under the model \eqref{eq: model}. To this end, we require some model assumptions:
\begin{itemize}
    \item Gaussian measurements: we assume that rows $\mPhi_i^T$ of $\mPhi$ are i.i.d. Gaussian vectors, i.e., $\mPhi_i \sim \NN(0, \mI_n)$. Note that the factor $\sqrt{m}$ in the model \eqref{eq: model} makes the columns of both $\mA$ and $\sqrt{m}\mI_m$ have the same scale, which helps our theoretical results to be more interpretable.

    %Note that if we assume that $\mPhi_i\sim\NN(0,\frac{1}{m}\mI_n)$ as in \cite{Foygel2014Corrupted, Jinchi2018Stable}, the factor $\sqrt{m}$ can be removed from \eqref{eq: model}.
    \item Unit norm of the signal: without loss of generality, we assume that $\|\vx^{\star}\|_2 = 1$ because the norm of $\vx^{\star}$ may be absorbed into the non-linear function $f$.
    \item Sub-Gaussian distribution of $\bar{\vy}_i=f_i(\ip{\mPhi_i}{\vx^{\star}})$: we assume that $\bar{\vy}_i=f_i(\ip{\mPhi_i}{\vx^{\star}})$ are sub-Gaussian variables as in \cite{Plan2017High}. To understand this assumption, note that $\ip{\mPhi_i}{\vx^{\star}}$ is Gaussian, $\bar{\vy}_i$ will be sub-Gaussian provided that $f$ does not grow faster than linearly, namely, $f(x) \leq a + b |x|$ for some scalars $a$ and $b$.
\end{itemize}

Under the above assumptions, we establish theoretical guarantees for $\TT$-Lasso \eqref{T_Lasso} under corrupted non-linear measurements \eqref{eq: model}. Our results demonstrate that under proper conditions, it is possible to disentangle signal and corruption in this quite challenging scenario.

%
%the contribution of this paper can be summarized as follows.
%\begin{itemize}
%	\item First, we establish the corrupted non-linear signal model which are widely used in practices and we illustrate that recovery can be done using an extension of $\TT$-Lasso.
%	\item Second, we present theoretical guarantee together with a geometric analysis of the extended $\TT$-Lasso procedure.
%\end{itemize}

\section{Preliminaries}\label{Preliminaries}
In this section, we review some preliminaries which underlie our analysis. Hereafter, $\S^{n-1}$ and $\B_2^n$  denote the unit sphere and ball in $\R^n$ under the $\ell_2$ norm respectively. We use the notation $C, C', c_1, c_2, \textrm{etc.},$ to refer to absolute constants whose value may change from line to line.

\subsection{Convex Geometry}
The \emph{tangent cone} of a set $\SS\subset \R^n$ at $\vx$ is defined as
\begin{equation*}\label{DefinitionofTangentCone}
\DD(\SS,\vx) = \{t\vu: t \ge 0, \vu \in \SS-\vx\}.
\end{equation*}
The tangent cone may also be called the descent cone.

The \emph{Gaussian width} and the \emph{Gaussian complexity} of a set $\SS \subset \R^{n}$ are, respectively, defined as
\begin{equation*}\label{Definition_Gaussian_width}
\omega(\SS) := \E \sup_{\vx \in \SS} \langle \vg, \vx \rangle, ~~ \textrm{where} ~~\vg\sim\NN(0,\mI_n),
\end{equation*}
and
\begin{equation*}\label{Definition_Gaussian_complexity}
\gamma(\SS) := \E \sup_{\vx \in \SS} |\langle \vg, \vx \rangle|, ~~ \textrm{where} ~~\vg\sim\NN(0,\mI_n).
\end{equation*}
These two geometric quantities are closely related to each other \cite{Liaw2017A}:
\begin{equation}\label{Relation}
\left(\omega(\SS)+\|\vy\|_2\right)/3 \leq \gamma(\SS) \leq 2(\omega(\SS)+\|\vy\|_2) ~~ \forall~\vy \in \SS.
\end{equation}

The \emph{local Gaussian width} of a set $\SS \subset \R^{n}$ is a function of parameter $t \ge 0$ defined as
\begin{equation*}\label{Definition_LGaussian_width}
\omega_t(\SS) := \E \sup_{\vx \in \SS \cap t\B_2^n} \langle \vg, \vx \rangle, ~~ \textrm{where} ~~\vg\sim\NN(0,\mI_n).
\end{equation*}

\subsection{High-Dimensional Probability}
A random variable $X$ is called a \emph{sub-Gaussian random variable} if the \emph{sub-Gaussian norm}
\begin{equation*}\label{Sub-Gaussian_Definition}
\|X\|_{\psi_2} = \inf\{t > 0: \E \exp(X^2/t^2) \leq 2\}
\end{equation*}
is finite. A random vector $\vx$ in $\R^n$ is \emph{sub-Gaussian random vector} if all of its one-dimensional marginals are sub-Gaussian random variables. The \emph{sub-Gaussian norm} of $\vx$ is defined as
\begin{equation*}
\|\vx\|_{\psi_2}:=\sup_{\vy\in\S^{n-1}}\big\| \ip{\vx}{\vy} \big\|_{\psi_2}.
\end{equation*}
A random vector $\vx$ in $\R^n$ is \emph{isotropic} if $\E(\vx\vx^T) = \mI_n$.

\subsection{A Useful Tool}

In the proofs of our main results, we make heavy use of the following matrix deviation inequality, which implies a tight lower bound for the restricted singular value of the extended sensing matrix $[\mA, \sqrt{m}\mI_m]$.
\begin{fact}[Extended Matrix Deviation Inequality, \cite{Jinchi2018Stable}]
	\label{Ext MDI}
	Let $\mA$ be an $m \times n$ matrix whose rows $\mA_i^T$ are independent centered isotropic sub-Gaussian vectors with $K = \max_i \|\mA_i\|_{\psi_2}$, and $\TT$ be a bounded subset of $\R^n\times\R^m$. Then
%	\begin{align*}
%	\E\sup_{(\va,\vb)\in \TT} &\left| \|\mA\va+\sqrt{m}\vb\|_2 - \sqrt{m}\cdot\sqrt{\|\va\|_2^2 + \|\vb\|_2^2} \right|\\
%	&\leq CK^2\cdot\gamma(\TT).
%	\end{align*}
	for any $s\geq 0$, the event
	\begin{align*}
	\sup_{ (\va,\vb)\in \TT }&\left| \|\mA\va+\sqrt{m}\vb\|_2 - \sqrt{m}\cdot\sqrt{\|\va\|_2^2 + \|\vb\|_2^2} \right|\\
	&\leq CK^2[ \gamma(\TT) + s\cdot \rad(\TT) ]
	\end{align*}
	holds with probability at least $1-\exp(-s^2)$, where $\rad(\TT) := \sup_{\vx \in \TT}\|\vx\|_2$ denotes the radius of $\TT$.
\end{fact}
In particular, when $\TT$ is a subset of $\S^{n+m-1}$ or $t\S^{n+m-1}$, Fact \ref{Ext MDI} implies that the event
\begin{equation} \label{LowerBound}
\inf_{(\va, \vb) \in \TT \cap \S^{n+m-1}}\|\mA \va+ \sqrt{m}\vb\|_2 \geq \sqrt{m} - CK^2{\gamma(\TT\cap \S^{n+m-1})}
\end{equation}
holds with probability at least $1-\exp \{-\gamma(\TT\cap \S^{n+m-1})^2\}$, or the event

\begin{equation}\label{LowerBound2}
\begin{split}
\inf_{(\va, \vb) \in \TT \cap t\S^{n+m-1}}&\left\|\mA \va+ \sqrt{m}\vb\right\|_2\\
&\geq t\sqrt{m} - CK^2{\gamma(\TT\cap t\S^{n+m-1})}
\end{split}
\end{equation}
holds with probability at least $1-\exp \{-\gamma(\TT\cap t\S^{n+m-1})^2/t^2\}$.
%============================================================================
\section{Main Results}\label{MainResults}

Before stating our result, we need to introduce two nonlinearity parameters, which are essentially the intrinsic mean and variance associated with the nonlinear map $f$. Let $g$ be a standard normal
random variable, the two parameters are defined as \cite{Plan2015The}:
\begin{align}\label{parameters}
  \textrm{Mean term}:       &~~ \mu = \E(f(g) \cdot g), \\
  \textrm{Variance term}:   &~~  \sigma^2 = \E (f(g) - \mu g)^2.
\end{align}

We then present two main results, one considers the case when the signal $(\mu \vx^{\star},\vv^{\star})$ lies at an extreme point of $\TT$, and the other assumes that $(\mu \vx^{\star},\vv^{\star})$ lies in the interior of $\TT$.

\begin{theorem}
	\label{them: Tangent_Cone}
	Let $(\hat{\vx}, \hat{\vv})$ be the solution to $\TT$-Lasso \eqref{T_Lasso}. Suppose that $\mPhi_i \sim \NN(0, \mI_n)$, $\vx^{\star} \in S^{n-1}$, and that $\bar{\vy}_i = f_i(\ip{\mPhi_i}{\vx^{\star}})$ are centered sub-Gaussian random variables with sub-Gaussian norm $\psi$. Assume that $(\mu \vx^{\star},\vv^{\star}) \in \TT$, and  let $\DD: = \DD(\TT, (\mu \vx^{\star},\vv^{\star}))$.
	If
	\begin{align}\label{NumberofMeasurements1}
	   m \geq C\cdot\omega_1(\DD)^2,
	\end{align}
	then, for any $0<s\leq\sqrt{m}$, the event
	\begin{align*}
	  \sqrt{\|\hat{\vx} - \mu \vx^{\star}\|_2^2+\|\hat{\vv} - \vv^{\star}\|_2^2} \leq \frac{C}{\sqrt{m}}\big(\omega_1(\DD)(\sigma+\psi+\mu)+s\sigma\big)
	\end{align*}
	holds with probability at least $1-2\exp(-cs^2\sigma^4/(\psi+\mu)^4)-\exp(-\gamma(\DD\cap\S^{n+m-1})^2)$.
\end{theorem}
%\begin{remark}[Anisotropic case]	
%\end{remark}
\begin{remark}[Relation to corrupted sensing]
If $f$ is the identity function, then we have $\mu=1, ~\sigma=0$, and $\psi=c$. Thus Theorem \ref{them: Tangent_Cone} implies that if $m\ge C\cdot\omega(\DD\cap\B_2^{n+m})^2$,  $\TT$-Lasso \eqref{T_Lasso} succeeds with high probability, which is consistent with the constrained recovery results in \cite[Theorem 1]{Foygel2014Corrupted} and \cite[Theorem 2]{Jinchi2018Stable}.
\end{remark}

Note that $\omega_1(\DD)^2$ is the \emph{effective dimension} of the descent cone $\DD$. When $(\mu\vx^{\star},\vv^{\star})$ lies on the boundary of $\TT$, which might lead to a narrow descent cone and hence a small effective dimension, then Theorem \ref{them: Tangent_Cone} becomes quite reasonable: a good estimation is guaranteed if the number of observations exceeds the effective dimension of $\DD$, which may be much smaller than the ambient dimension $n+m$. However, when $(\mu\vx^{\star},\vv^{\star})$ is an interior point of $\TT$, the descent cone is the entire space, the effective dimension $\omega_1(\DD)^2$ is of the order of the ambient dimension $n+m$. In this case, the results in Theorem \ref{them: Tangent_Cone} become meaningless. The following theorem deals with this situation. As it turns out that local Gaussian width serves as a new measure to characterize the low dimension structure of set $\TT$ which is unnecessary to be a cone.

\begin{theorem}
	\label{them: noTangent_Cone}
	Let $(\hat{\vx}, \hat{\vv})$ be the solution to $\TT$-Lasso \eqref{T_Lasso}. Suppose that $\mPhi_i \sim \NN(0, \mI_n)$, $\vx^{\star} \in S^{n-1}$, and that $\bar{\vy}_i = f_i(\ip{\mPhi_i}{\vx^{\star}})$ are centered sub-Gaussian random variables with sub-Gaussian norm $\psi$. Assume that $(\mu \vx^{\star},\vv^{\star}) \in \TT$ and let $\KK:=\TT-(\mu \vx^{\star},\vv^{\star})$ is a star shaped set\footnote{$\KK$ is a star shaped set if it satisfies $\lambda \KK \subset \KK$ for any $0 \leq \lambda \leq 1$. Specially, any convex set containing origin is star shaped.}.
	If
	\begin{align}\label{NumberofMeasurements2}
	m \geq C\cdot\omega_t(\KK)^2/t^2,
	\end{align}
	then, for any $t>0,~0<s\leq\sqrt{m}$, the event
	\begin{align*}
	&\sqrt{\|\hat{\vx} - \mu \vx^{\star}\|_2^2+\|\hat{\vv} - \vv^{\star}\|_2^2}\\
	&\qquad\qquad \leq t+\frac{C}{\sqrt m}\left(\frac{\omega_t(\KK)(\sigma+\psi+\mu)}{t} + s\sigma \right)
	\end{align*}
	holds with probability at least $1-2\exp(-cs^2\sigma^4/(\psi+\mu)^4)-\exp(-\gamma(\KK\cap t\S^{n+m-1})^2/t^2)$.
\end{theorem}
\begin{remark}[Local Gaussian width]
Note that if we let $t \to 0$, then ${\omega_t(\KK)}/{t}$ goes to $\omega(\B_2^{n+m})$, which is of the order of $\sqrt{n+m}$. Then the results in Theorem \ref{them: noTangent_Cone} are exact what in Theorem \ref{them: Tangent_Cone} when $(\mu\vx^{\star},\vv^{\star})$ is an interior point of $\TT$. This suggests that Theorem \ref{them: Tangent_Cone} can be regarded as an extreme case of Theorem \ref{them: noTangent_Cone}, and local Gaussian width can better characterizes the low dimension structure of sets than Gaussian width.
\end{remark}
\begin{remark}[Relation to results in \cite{Plan2015The}]
Theorems \ref{them: Tangent_Cone} and  \ref{them: noTangent_Cone} show that the recovery error can be diminished to an arbitrarily small degree provided that the number of measurements is large enough. Specially, in the corruption-free case (i.e., without the $\psi+\mu$ term in the high-probability bounds), our results also agree with Theorem $1.4$ and Theorem $1.9$ in \cite{Plan2015The}.
\end{remark}
\section{Proofs of Main Results}\label{proof of main}

Before proving Theorems \ref{them: Tangent_Cone} and  \ref{them: noTangent_Cone}, we require two useful lemmas.
\begin{lemma}\label{lemma: UI}
	Suppose that $\mPhi_i \sim \NN(0, \mI_n)$ and $\bar{\vy}_i=f_i(\ip{\mPhi_i}{\vx^{\star}})$ are centered sub-Gaussian random variables with sub-Gaussian
	norm $\psi$. Assume $\KK^t=\KK_{\va}^t\times\KK_{\vb}^t \subset t\B_2^{n+m}$ is a star shaped set and let $\vz:=f(\mPhi\vx^{\star})-\mPhi\mu\vx^{\star}$. Then,
%	\begin{align*}
%	\E\sup_{(\va,\vb)\in \TT} \ip{\mPhi\va + \sqrt{m}\vb}{\vz} \leq C\sqrt{m}\big[\omega(\TT)(\sigma+\psi+\mu)+t\eta\big].
%	\end{align*}
	for any $0<s \leq \sqrt{m}$, the event
	\begin{align*}
	\sup_{(\va,\vb)\in \KK^t} \ip{\mPhi\va + \sqrt{m}\vb}{\vz} \leq C\sqrt{m}\big[\omega(\KK^t)(\sigma+\psi+\mu)+st\sigma\big]
	\end{align*}
	holds with probability at least $1-2\exp(-c{s^2\sigma^4}/{(\psi+\mu)^4})$.
\end{lemma}

\begin{proof}
See Appendix \ref{ProofofUI}.
\end{proof}

\begin{lemma} \label{lem: large scale conditioning}
		Let $\KK= \TT-(\mu \vx^{\star},\vv^{\star})$ be a star shaped set and $t > 0$.
		Suppose that $m \geq C\cdot \omega_t(\KK)^2/t^2$.
		Then, the following lower bound
		\begin{align*}
		\|\mPhi\vh+\sqrt{m}\ve\|_2 \geq \frac{\sqrt{m}}{2} \sqrt{\|\vh\|_2^2 + \|\ve\|_2^2}
		\end{align*}
		holds for all $(\vh,\ve) \in \KK$ satisfying $\sqrt{\|\vh\|_2^2 + \|\ve\|_2^2} \geq t$ with probability at least $1-\exp \big(-\gamma(\KK\cap t\S^{n+m-1})^2/t^2\big)$.
\end{lemma}	
\begin{proof}
		Let $\lambda =\frac{t}{\sqrt{\|\vh\|_2^2 + \|\ve\|_2^2}}\leq 1$ and $(\vu,\vv)=\lambda\cdot(\vh,\ve)$. Then $(\vu,\vv)\in \lambda\KK\cap t S^{n+m-1}$. Thus we have
		\begin{align*}
		&\inf_{(\vh,\ve) \in \KK, ~\|(\vh,\ve)\|_2 \geq t} \frac{\|\mPhi \vh+\sqrt{m}\ve\|_2}{\sqrt{\|\vh\|_2^2 + \|\ve\|_2^2}} \\
		&\qquad\qquad= \inf_{(\vu,\vv)\in \lambda\KK\cap t S^{n+m-1}} \frac{\|\mPhi \vu+\sqrt{m}\vv\|_2}{t}\\
		&\qquad\qquad\geq \inf_{(\vu,\vv)\in \KK\cap t S^{n+m-1}} \frac{\|\mPhi \vu+\sqrt{m}\vv\|_2}{t}\\
		&\qquad\qquad \geq \sqrt{m} - C'{\gamma(\KK\cap t\S^{n+m-1})}/t\\
		&\qquad\qquad \geq \sqrt{m} - C''{\omega_t(\KK)}/t\\
		&\qquad\qquad  \geq \frac{{\sqrt m}}{2}
		\end{align*}
holds with probability at least $1-\exp \big(-\gamma(\KK\cap t\S^{n+m-1})^2/t^2\big)$. The first inequality holds because $\KK$ is star shaped, then $\lambda\KK\subset\KK$. The second inequality follows from \eqref{LowerBound2}. The third inequality holds because \eqref{Relation} and $\vzero\in \KK$, i.e.,
	\begin{align*}
	\gamma(\KK\cap t\S^{n+m-1})\leq\gamma(\KK\cap t\B_2^{n+m})
	\leq 2\omega(\KK\cap t\B_2^{n+m}).
	\end{align*}
The last inequality follows from the assumption on the number of measurements $m \geq C\cdot \omega_t(\KK)^2/t^2$.
\end{proof}

\subsection{Proof of Theorem \ref{them: Tangent_Cone}}

\begin{proof}
	For clarity, the proof is divided into three steps.

	\textbf{Step 1: Problem reduction.} Since $(\hat{\vx}, \hat{\vv})$ is the solution to the $\TT$-Lasso problem \eqref{T_Lasso} and $(\mu \vx^{\star},\vv^{\star}) \in \TT$, then we have
	\begin{align}\label{reduction1}
	\|\vy-\mPhi\hat{\vx}-\sqrt{m}\hat{\vv}\|_2 \le \|\vy-\mPhi\mu\vx^\star-\sqrt{m}\vv^\star\|_2.
	\end{align}
	Recall that $\vz=f(\mPhi\vx^\star)-\mPhi\mu\vx^\star$, then $\vy=\mPhi\mu\vx^\star+\sqrt{m}\vv^\star+\vz$. Let $\vh=\hat{\vx}-\mu\vx^\star$ and $\ve=\hat{\vv}-\vv^\star$. Then \eqref{reduction1} can be reformulated as
	\begin{align}\label{reduction2}
	\|\mPhi\vh+\sqrt{m}\ve-\vz\|_2 \le \|\vz\|_2.
	\end{align}
	Squaring both sides of \eqref{reduction2} yields
	\begin{align}\label{reduction3}
	\left\| {\mPhi \vh + \sqrt m \ve} \right\|_2^2 \le 2\ip{\mPhi \vh + \sqrt m \ve}{\vz}.
	\end{align}
	
	\textbf{Step 2: Lower Bound on $\| {\mPhi \vh + \sqrt m\ve}\|_2$.} Define the error set
     \begin{align*}
     \EE(\mu \vx^{\star},\vv^{\star}):&=\{(\vh,\ve)\in \R^n\times\R^m: (\mu \vx^{\star}+\vh,\vv^{\star}+\ve)\in \TT \}\\
     &=\TT-(\mu \vx^{\star},\vv^{\star}),
    \end{align*}
     in which the error vector $(\hat{\vx}-\mu \vx^{\star},\hat{\vv}-\vv^{\star})$ lives. Clearly, $\EE(\mu \vx^{\star},\vv^{\star})$ belongs to the tangent cone $\DD(\TT,(\mu \vx^{\star},\vv^{\star}))$.
     It then follows from \eqref{LowerBound} that the event
	\begin{align*}
	&\|\mPhi \vh+ \sqrt m \ve\|_2\\
	&\quad= \sqrt{\|\vh\|_2^2+\|\ve\|_2^2} \cdot \left\|\frac{\mPhi\vh}{\sqrt{\|\vh\|_2^2+\|\ve\|_2^2}}+\frac{\sqrt{m}\ve}{\sqrt{\|\vh\|_2^2+\|\ve\|_2^2}}\right\|_2  \\
	&\quad \geq \sqrt{\|\vh\|_2^2+\|\ve\|_2^2} \cdot (\sqrt{m} - C{\gamma(\DD\cap \S^{n+m-1})})\\
	&\quad \geq \sqrt{\|\vh\|_2^2+\|\ve\|_2^2} \cdot (\sqrt{m} - C_1{\omega_1(\DD)})\\
	&\quad \geq \frac{{\sqrt m}}{2} \sqrt{\|\vh\|_2^2+\|\ve\|_2^2}
	\end{align*}
	holds with probability at least $1-\exp \{-\gamma(\DD\cap \S^{n+m-1})^2\}$. The second inequality holds because \eqref{Relation} and $\vzero\in \DD$, namely
	\begin{align*}
	\gamma(\DD\cap \S^{n+m-1})\leq\gamma(\DD\cap \B_2^{n+m})
	\leq 2\omega(\DD\cap \B_2^{n+m}).
	\end{align*}
	The last inequality is due to \eqref{NumberofMeasurements1}.
	
	\textbf{Step 3: Upper Bound on $\ip{\mPhi \vh + \sqrt m\ve}{\vz}$.} It follows Lemma \ref{lemma: UI} that (by setting $t=1$) the event
	\begin{align*}
	&\ip{\mPhi \vh + \sqrt m\ve}{\vz}\\
	& = \sqrt{\|\vh\|_2^2+\|\ve\|_2^2} \ip{\frac{\mPhi\vh+\sqrt{m}\ve}{\sqrt{\|\vh\|_2^2+\|\ve\|_2^2}}}{\vz}\\
	%&\leq C \sqrt{\|\vh\|_2^2+\|\ve\|_2^2} \cdot\sqrt{m}\big[\omega(D\cap \S^{n+m-1})(\sigma+\psi+\mu)+s\sigma\big]\\
	%&\leq C \sqrt{\|\vh\|_2^2+\|\ve\|_2^2} \cdot\sqrt{m}\big[\omega(D\cap \B_2^{n+m})(\sigma+\psi+\mu)+s\sigma\big]\\
	&\leq C\sqrt{m} \sqrt{\|\vh\|_2^2+\|\ve\|_2^2} \cdot \big[\omega_1(\DD)(\sigma+\psi+\mu)+s\sigma\big]
	\end{align*}
	holds with probability at least $1-2\exp(-cs^2\sigma^4/{(\psi+\mu)^4})$.
	
	Putting everything together and taking union bound, we have that, with probability at least $1-2\exp(-cs^2\sigma^4/{(\psi+\mu)^4})-\exp \big(-\gamma(\DD\cap\S^{n+m-1})^2\big)$,
	\begin{align*}
	&\frac{m}{4}(\|\vh\|_2^2+\|\ve\|_2^2)\\
	&\qquad \leq C\sqrt{m} \sqrt{\|\vh\|_2^2+\|\ve\|_2^2} \cdot \big[\omega_1(\DD)(\sigma+\psi+\mu)+s\sigma\big].
	\end{align*}
	Rearranging completes the proof of Theorem \ref{them: Tangent_Cone}.
\end{proof}

\subsection{Proof of Theorem \ref{them: noTangent_Cone}}

\begin{proof}
First note that if $\sqrt{\|\vh\|_2^2 + \|\ve\|_2^2} \leq  t$, then Theorem \ref{them: noTangent_Cone} holds trivially. So it is sufficient to prove Theorem \ref{them: noTangent_Cone} under assumption $\sqrt{\|\vh\|_2^2 + \|\ve\|_2^2} \geq t$.

Similar to Step 1 of the proof of Theorem \ref{them: Tangent_Cone}, we have
	\begin{equation}\label{eq: h via ip}
	\left\| {\mPhi \vh + \sqrt m \ve} \right\|_2^2 \le 2\ip{\mPhi \vh + \sqrt m \ve}{\vz}.
	\end{equation}

Observe that the error vector $(\vh,\ve)$ belongs to a star shaped set, namely $\KK=\TT-(\mu \vx^{\star},\vv^{\star})$. It then follows from Lemma \ref{lem: large scale conditioning} that the following event
		\begin{align}\label{proof1}
		\|\mPhi\vh+\sqrt{m}\ve\|_2 \geq \frac{\sqrt{m}}{2} \sqrt{\|\vh\|_2^2 + \|\ve\|_2^2}
		\end{align}
holds with probability at least $1-\exp \big(-\gamma(\KK\cap t\S^{n+m-1})^2/t^2\big)$.

Combining \eqref{eq: h via ip} and \eqref{proof1} yields
	\begin{equation}\label{eq: hez}				
	\frac{m}{4}(\|\vh\|_2^2 + \|\ve\|_2^2) \leq 2\ip{\mPhi\vh+\sqrt{m}\ve}{\vz}.
	\end{equation}

	Note that $\sqrt{\|\vh\|_2^2 + \|\ve\|_2^2}\ge t$, we cannot use the upper bound in Lemma \ref{lemma: UI} directly. So dividing both sides of \eqref{eq: hez} by $m\delta=m \sqrt{\|\vh\|_2^2 + \|\ve\|_2^2}$, we obtain
	\begin{align*}
	\sqrt{\|\vh\|_2^2 + \|\ve\|_2^2} &\leq \frac{8}{m} \delta^{-1}\ip{\mPhi\vh+\sqrt{m}\ve}{\vz}\\
	&\leq \frac{8}{m} \sup_{(\vu,\vv) \in \delta^{-1} \KK\cap \B_2^{n+m}} \ip{\mPhi\vu+\sqrt{m}\vv}{\vz}\\
	&\leq \frac{8}{m} \sup_{(\vu,\vv) \in t^{-1}\KK\cap \B_2^{n+m}} \ip{\mPhi\vu+\sqrt{m}\vv}{\vz}\\
	&=\frac{8}{mt} \sup_{(\va,\vb) \in \KK\cap t\B_2^{n+m}} \ip{\mPhi\va+\sqrt{m}\vb}{\vz}\\
	&\leq\frac{C}{\sqrt m}\big[\frac{\omega_t(\KK)(\sigma+\psi+\mu)}{t} + s\sigma \big]
	\end{align*}
	holds with probability at least $1-2\exp(-cs^2\sigma^4/{(\psi+\mu)^4})-\exp \big(-\gamma(\KK\cap t\S^{n+m-1})^2/t^2\big)$. In the second inequality we set $(\vu,\vv)=\delta^{-1}(\vh,\ve)$. The third inequality holds due to $\KK$ is star shaped, namely $t\delta^{-1}\KK\subset\KK$ and hence $\delta^{-1}\KK\subset t^{-1}\KK$. In the fourth line we let $(\va,\vb)=t(\vu,\vv)$. The last inequality follows from Lemma \ref{lemma: UI}. Thus we complete the proof.
\end{proof}

%\section{Numerical Simulations}\label{Simulations}
\section{Conclusion}\label{Conclusion}
In this paper, we have analyzed performance guarantees for $\TT$-Lasso which is used to recover a structured signal from corrupted non-linear Gaussian measurements. The theoretical results may be of help in some practical applications such as dealing with saturation error in quantization which has been a challenge in the area of signal processing. As for future work, it is worthwhile to deduce the explicit expressions of the main results for different specific problems, and to consider penalized recovery procedures rather than a constrained one for computational purposes.
\appendices

\section{Proof of Lemma \ref{lemma: UI}}\label{ProofofUI}

\subsection{Auxiliary Definitions and Facts}

To prove Lemma \ref{lemma: UI}, we require some additional definitions and facts.

\begin{definition}[Sub-exponential random variable and vector]
A random variable $X$ is called a \emph{sub-exponential random variable} if the \emph{sub-exponential norm}
\begin{equation*}\label{Sub-exponential_Definition}
\|X\|_{\psi_1} = \inf\{t > 0: \E \exp(\left|X\right|/t) \leq 2\}
\end{equation*}
is finite. A random vector $\vx$ in $\R^n$ is called \emph{sub-exponential random vector} if all of its one-dimensional marginals are sub-exponential random variables. The \emph{sub-exponential norm} of $\vx$ is defined as
\begin{equation*}
\|\vx\|_{\psi_1}:=\sup_{\vy\in\S^{n-1}}\big\| \ip{\vx}{\vy} \big\|_{\psi_1}.
\end{equation*}
\end{definition}

\begin{fact}[Sub-Gaussian distributions with independent coordinates]\cite[Lemma 3.4.2]{Vershynin2018}
	\label{pro: sub_G vector}
	Let $X=(X_1, \ldots, X_n)^T\in\R^n$ be a random vector with independent, mean zero, sub-Gaussian coordinates $X_i$. Then $X$ is a sub-Gaussian random vector, and
	\begin{equation*}
	\|X\|_{\psi_2} \leq C\max_{i\leq n} \|X_i\|_{\psi_2}.
	\end{equation*}
\end{fact}

\begin{fact}[Product of sub-Gaussian is sub-exponential]\cite[Lemma 2.7.7]{Vershynin2018} \label{Product of subgaus is subexp}
	Let $X$ and $Y$ be sub-Gaussian random variables (not necessarily independent). Then $XY$ is sub-exponential. Moreover,
	\begin{align*}
	\|XY\|_{\psi_1} \leq \|X\|_{\psi_2}\|Y\|_{\psi_2}.
	\end{align*}
\end{fact}

\begin{fact}[Centering]\cite[Lemma 2.6.8 and Exercise 2.7.10]{Vershynin2018}
	\label{pro: center}
	If $X$ is sub-Gaussian (or sub-exponential), then so is $X-\E X$. Moreover,
	\begin{equation*}
	\|X - \E X\|_{\psi_2} \leq C \|X\|_{\psi_2} ~~\textrm{and}~~\|X-\E X\|_{\psi_1} \leq C \|X\|_{\psi_1}.
	\end{equation*}
\end{fact}

\begin{fact}[Bernstein-type inequality] \cite[Theorem 2.8.2]{Vershynin2018}
	\label{Bernstein ineq}
	Let $X_1, X_2, \ldots, X_m $ be independent, mean-zero, sub-exponential random variables, and $\va = (a_1, a_2, \ldots, a_m)^{T}\in\R^{m}$. Then, for any $t \geq 0$, we have
	\begin{align*}
	&\Pr{ \left| \sum_{i=1}^{m} a_{i} X_{i} \right| \geq t }\\
	&\qquad \leq 2 \exp \left\{ -c \min  \left( \frac{t^2}{K^2 \|\va\|_2^2}, \frac{t}{K\|\va\|_{\infty} } \right)\right\},
	\end{align*}
	where $K = \max_{i} \|X_{i}\|_{\psi_1}$.
\end{fact}

\begin{fact}[Gaussian concentration]\cite[Theorem 5.2.2]{Vershynin2018}
	\label{gaussian concentration}
	Consider a random vector $X\sim \NN(0,\mI_n)$ and a Lipschitz function $f:~\R^n\to\R$ with Lipschitz norm $\|f\|_{\textrm{Lip}}$ (with respect to the Euclidean metric). Then for any $t \geq 0$, we have
	\begin{equation*}
     \Pr{ |f(X) - \E f(X)| \geq t} \leq 2 \exp\left(\frac{-ct^2}{\|f\|_{\textrm{Lip}}^2}\right).
	\end{equation*}
\end{fact}

\begin{fact}[Talagrand's Majorizing Measure Theorem] \cite[Theorem 2.2.27]{talagrand2014upper} or \cite[Theorem 8]{Liaw2017A}
	\label{Talagrand Them}
	Let $( X_{\vu} )_{\vu\in \SS}$ be a random process indexed by points in a bounded set $\SS \subset \R^{n}$. Assume that the process has sub-Gaussian increments, that is, there exists $M \geq 0$ such that
	\begin{equation*}
	\| X_{\vu} - X_{\vv} \|_{\psi_2} \leq M \|\vu-\vv\|_2 ~~~~ \text{for every} ~~ \vu,\vv \in \SS.
	\end{equation*}
	Then,
%	\begin{equation*}\label{Expectation_Bound}
%	\E \sup_{\vu,\vv \in \TT} \big| X_{\vu} - X_{\vv} \big| \leq C M \omega(\TT).
%	\end{equation*}
	for any $s\geq 0$, the event
	\begin{equation*}\label{High_Probability_Bound}
	\sup_{\vu, \vv \in \SS} \big|X_{\vu} - X_{\vv}\big| \leq CM \big[ \omega(\SS) + s \cdot \diam(\SS) \big]
	\end{equation*}
	holds with probability at least $1- \exp(-s^2)$, where $\diam(\SS) := \sup_{\vx,\vy\in \SS}\|\vx-\vy\|_2$ denotes the diameter of $\SS$.
\end{fact}

\subsection{Proof of Lemma \ref{lemma: UI}}

We are now in position to prove Lemma \ref{lemma: UI}. Observe that
$$\sup_{(\va,\vb) \in \KK^t} \ip{\mPhi\va + \sqrt{m}\vb}{\vz} \leq \sup_{\va \in \KK_{\va}^t} \ip{\mPhi\va}{\vz} + \sqrt{m} \sup_{\vb\in \KK_{\vb}^t} {\ip{\vb}{\vz}}.$$
So it suffices to bound the two terms on the right side. To this end, we have the following two lemmas.

\begin{lemma}\label{ip_bound}
	Under the settings of Lemma \ref{lemma: UI}, then for any $0<s \leq \sqrt{m}$, the event
	\begin{align*}
	\sup_{\va \in \KK_{\va}^t} \ip{\mPhi\va}{\vz} \leq C\sqrt{m}\big[\omega(\KK_{\va}^t)\sigma+st\sigma\big]
	\end{align*}
	holds with probability at least $1-2\exp(-{cs^2\sigma^4}/{(\psi+\mu)^4})$.
\end{lemma}
\begin{proof}
See Appendix \ref{proof of ip_bound}.
\end{proof}

\begin{lemma}\label{sub-G}
	Under the settings of Lemma \ref{lemma: UI}, the event
	\begin{align*}
	\sup_{\vb\in \KK_{\vb}^t}  \ip{\vb}{\vz} &\leq C[(\psi+\mu)\omega(\TT_{\vb})+st\sigma]
	\end{align*}
	holds with probability at least $1-\exp(-\frac{s^2\sigma^2}{(\psi+\mu)^2})$.
\end{lemma}

\begin{proof}
	Note that $\vz_i$ are i.i.d. centered sub-Gaussian variables with $\psi_2$-norm
	\begin{align}\label{z_sub_G}
	  K & := \|\vz_i\|_{\psi_2}  = \|f_i(\ip{\mPhi_i}{\vx^{\star}})-\mu\ip{\mPhi_i}{\vx^{\star}}\|_{\psi_2}  \\ \notag
      & \leq \psi+ C_1 \mu \leq C_2(\psi+\mu).
	\end{align}
    Then by Fact \ref{pro: sub_G vector}, $\vz$ is a sub-Gaussian random vector with
    $$\|\vz\|_{\psi_2} \leq  C_3 (\psi+\mu).$$
    Define the random process $X_{\vb}:=\ip{\vb}{\vz}$, which has sub-Gaussian increments:
	\begin{align*}
	\|X_{\vb} - X_{\vb'}\|_{\psi_2}
	&= \left\|\ip{\vz}{\vb-\vb'}\right\|_{\psi_2} \\
	&= \|\vb-\vb'\|_2 \cdot \left\|\ip{\vz}{\frac{\vb-\vb'}{\|\vb-\vb'\|_2}}\right\|_{\psi_2} \\
	&\le C_3(\psi+\mu)\|\vb-\vb'\|_2.
    \end{align*}
    Note that $\vzero \in \KK_{\vb}^t$, it then follows from Talagrand's Majorizing Measure Theorem (Fact \ref{Talagrand Them}) that the event
 	\begin{align*}
	\sup_{\vb\in \KK_{\vb}^t} \ip{\vb}{\vz} &\leq \sup_{\vb\in \KK_{\vb}^t} \left|\ip{\vb}{\vz}\right|\\
	 &\leq C_4(\psi+\mu)\big[\omega(\KK_{\vb}^t) + u\cdot \diam(\KK_{\vb}^t)\big]\\
     &\leq C_5(\psi+\mu)\big[\omega(\KK_{\vb}^t)+ut\big]
\end{align*}
    holds with probability at least $1- \exp(-u^2)$. The last inequality holds because $\diam(\KK_{\vb}^t) = \sup_{\vx,\vy\in \KK_{\vb}^t}\|\vx-\vy\|_2\leq 2t$.
    Setting $u = \frac{s\sigma}{\psi+\mu}$ yields the desired results.
\end{proof}

Thus, combing Lemma \ref{ip_bound} and Lemma \ref{sub-G} yields the proof of Lemma \ref{lemma: UI}, namely, for any $0<s \leq \sqrt{m}$, the event
	\begin{align*}
	&\sup_{(\va,\vb)\in \KK^t} \ip{\mPhi\va + \sqrt{m}\vb}{\vz}\\
	&\qquad\leq \sup_{\va \in \KK_{\va}^t} \ip{\mPhi\va}{\vz} + \sqrt{m}\sup_{\vb\in \KK_{\vb}^t} {\ip{\vb}{\vz}}\\
	&\qquad\leq C_6\sqrt{m}\big[\omega(\KK_{\va}^t)\sigma+st\sigma + (\psi+\mu)\omega(\KK_{\vb}^t)+st\sigma\big]\\
	&\qquad\leq C_7\sqrt{m}\big[\omega(\KK^t)(\sigma+\psi+\mu)+st\sigma\big]
\end{align*}
	holds with probability at least
	\begin{align*}
	&1-2\exp(-{cs^2\sigma^4}/{(\psi+\mu)^4})-\exp(-s^2\sigma^2/(\psi+\mu)^2)\\
	&\ge 1-2\exp(-{c's^2\sigma^4}/{(\psi+\mu)^4}).
	\end{align*}
	In the last inequality we have used the facts that $\omega(\KK_{\va}^t) \leq \omega(\KK^t)$ and $\omega(\KK_{\vb}^t) \leq \omega(\KK^t)$.
%	\begin{align*}
%	&\omega(\KK_{\va}^t)+\omega(\KK_{\vb}^t)\\
%    & =\E\sup_{\va \in\KK_{\va}^t\cap t\B_2^n}\ip{\vg}{\va}+\E\sup_{\vb \in\KK_{\vb}^t\cap t\B_2^m}\ip{\vh}{\vb}\\
%    & = \sqrt{2}\E\sup_{\vu\in\frac{1}{\sqrt 2}\KK_{\va}^t\cap \frac{t}{\sqrt 2}\B_2^n}\ip{\vg}{\vu}+\sqrt{2}\E\sup_{\vv\in\frac{1}{\sqrt 2}\KK_{\vb}^t\cap \frac{t}{\sqrt 2}\B_2^m}\ip{\vh}{\vv}\\
%    & =\sqrt{2}\E\sup_{\vu\in\frac{1}{\sqrt 2}\KK_{\va}^t\cap \frac{t}{\sqrt 2}\B_2^n, \vv\in\frac{1}{\sqrt 2}\KK_{\vb}^t\cap \frac{t}{\sqrt 2}\B_2^m}\ip{\vg}{\vu}+\ip{\vh}{\vv}\\
%	& \leq \sqrt{2}\E\sup_{(\vu,\vv)\in\frac{1}{\sqrt 2}\KK\cap t\B_2^{n+m}}\ip{\vg}{\vu}+\ip{\vh}{\vv}\\
%	&\leq \sqrt{2}\E\sup_{(\vu,\vv)\in\KK\cap t\B_2^{n+m}}\ip{\vg}{\vu}+\ip{\vh}{\vv}\\
%	&=\sqrt{2}\omega(\KK),
%	\end{align*}
%    where $\vg$ and $\vh$ are independent standard Gaussian vectors. The last inequality holds because $\KK$ is star shaped. Thus we complete the proof.

\section{Proof of Lemma \ref{ip_bound}}\label{proof of ip_bound}
The proof of Lemma \ref{ip_bound} is inspired by \cite{Plan2017High}. For clarity, the proof is divided into the following three steps.
	 %First, we introduce the orthogonal decomposition of Gaussian variables. We then decomposition the inner product into two parts using Gaussian decomposition. Finally, we bound the two parts respectively.\\
	
\textbf{Step 1: Problem Reduction.}
Since $\vz_i$ are not independent of $\mPhi_i$, to facilitate the analysis, we need to ``decouple'' them as much as possible. To this end, we consider the orthogonal decomposition of the vectors $\mPhi_i$ along the direction of $\vx^{\star}$ and its orthogonal complementary space. More precisely, we express
\begin{equation}\label{eq: orthogonal decomposition}
	\mPhi_i = \mP\mPhi_i+\mP^{\perp}\mPhi_i=\ip{\mPhi_i}{\vx^{\star}} \vx^{\star}+ \mP^{\perp}\mPhi_i, %\ip{\mPhi_i}{\vx^{\star}} \vx^{\star} + \mPhi_i^\perp.
\end{equation}
where $\mP :=\vx^{\star}{\vx^{\star}}^{\perp}$ and $\mP^{\perp} :=\mI_n-\mP$.	
Thus we have
\begin{align*}
	&\sup_{\va \in \KK_{\va}^t} \ip{\mPhi\va}{\vz} = \sup_{\va \in \KK_{\va}^t} \sum_{i=1}^{m} \vz_i\ip{\mPhi_i}{\va}\\
	&= \sup_{\va \in \KK_{\va}^t} \sum_{i=1}^{m} \vz_i\ip{\ip{\mPhi_i}{\vx^{\star}} \vx^{\star} + \mP^\perp\mPhi_i}{\va}\\
	&\leq \sup_{\va \in \KK_{\va}^t} \sum_{i=1}^{m} \ip{\vz_i\ip{\mPhi_i}{\vx^{\star}} \vx^{\star}}{\va} + \sup_{\va \in \KK_{\va}^t} \sum_{i=1}^{m} \ip{\vz_i \mP^\perp\mPhi_i}{\va}\\
	&\leq \left|\sum_{i=1}^{m} \vz_i\ip{\mPhi_i}{\vx^{\star}}\right|\sup_{\va \in \KK_{\va}^t} \left|\ip{\vx^{\star}}{\va}\right| + \sup_{\va \in \KK_{\va}^t} \sum_{i=1}^{m} \ip{\vz_i \mP^\perp\mPhi_i}{\va}\\
	&\leq \left|\sum_{i=1}^{m} \vz_i\ip{\mPhi_i}{\vx^{\star}}\right|\cdot t + \sup_{\va \in \KK_{\va}^t} \sum_{i=1}^{m} \ip{\vz_i \mP^\perp\mPhi_i}{\va}\\
	&:= E_1+E_2.
\end{align*}

\textbf{Step 2: Bound $E_1$.}
Define $\xi_i := \vz_i\ip{\mPhi_i}{\vx^{\star}} =  \big[ f(\ip{\mPhi_i}{\vx^{\star}}) - \mu \ip{\mPhi_i}{\vx^{\star}} \big] \ip{\mPhi_i}{\vx^{\star}}$. By the definition of $\mu$, it is not hard to check that $\E\xi_i=0$.
Note that $\vz_i$ have sub-Gaussian norm $K \leq C_2(\psi+\mu)$ (see \eqref{z_sub_G}) and $\ip{\mPhi_i}{x^{\star}} \sim \NN(0,1)$. It then follows from Fact \ref{Product of subgaus is subexp} that $\xi_i$ are i.i.d. centered sub-exponential variables with $\|\xi_i\|_{\psi_1}=C'K$. Let $\epsilon=s/\sqrt m \leq 1$. A Bernstein-type inequality (Fact \ref{Bernstein ineq}) implies that
\begin{align*} \label{eq: E1 whp}
	\Big| \frac{1}{m} \sum_{i=1}^m \xi_i \Big| \leq \epsilon \sigma, \quad \text{and thus} \quad E_1 \leq m\epsilon\sigma t
\end{align*}
holds with probability at least
\begin{align*}
	1 - 2 \exp \Big[ -c \min \Big( \frac{\epsilon^2 \sigma^2}{K^2}, \, \frac{\epsilon \sigma}{K} \Big) m \Big]
	\ge 1 - 2 \exp \Big( -\frac{c m \epsilon^2 \sigma^2}{K^2} \Big).
\end{align*}
In the last inequality we have used the facts that $\sigma^2 = \E \vz_i^2 \leq C K^2$ and $\epsilon \leq 1$.
	
\textbf{Step 3: Bound $E_2$.} Let $\vw = \sum_{i=1}^{m}\vz_i\mP^{\perp}\mPhi_i$. By the orthogonal decomposition \eqref{eq: orthogonal decomposition}, $\mP^{\perp}\mPhi_i$ and $\vz_i$ are independent \cite[Lemma 8.1]{Plan2017High}. Fixing $\vz_i$, a direct calculation shows that
	\begin{align*}
	\vw\sim k\cdot\NN(0,\mP^{\perp}),
	\end{align*}
	where $k= \sqrt{{\sum_{i=1}^{m} \vz_i^2}}$. Thus, conditioning on $\vz_i$, $E_2=\sup_{\va \in \KK_{\va}^t}\ip{\vw}{\va}=k\cdot\sup_{\va \in \KK_{\va}^t}\ip{\mP^{\perp}\vg}{\va}$.
	
	Note that $\vz_i^2$ are sub-exponential variables with mean $\sigma^2$ and $\psi_1$-norm $CK^2$. By Fact \ref{pro: center}, $\vz_i^2 - \sigma^2$ are centered sub-exponential variables with $\psi_1$-norm $C'K^2$. A similar application of Bernstein-type inequality (Fact \ref{Bernstein ineq}) yields that
	\begin{align*}
	\Big| \frac{1}{m} \sum_{i=1}^m (\vz_i^2-\sigma^2) \Big| \leq 3\sigma^2, \quad \text{and thus} \quad k^2 \leq 4m\sigma^2
	\end{align*}
	holds with probability at least
	\begin{align*}
	1 - 2 \exp \Big[ -c \min \Big( \frac{\sigma^4}{K^4}, \, \frac{\sigma^2}{K^2} \Big) m \Big]
	\ge 1 - 2 \exp \Big( -\frac{c m \sigma^4}{K^4} \Big).
	\end{align*}
	Here we have used the fact that $\sigma^2 = \E \vz_i^2 \leq C K^2$ again. Therefore, with probability at least $1 - 2 \exp ( -{c m \sigma^4}/{K^4} )$,
    $$E_2 \leq 2\sqrt{m}\sigma \cdot\sup_{\va \in \KK_{\va}^t}\ip{\mP^{\perp}\vg}{\va}.$$
	
	We next bound $\sup_{\va \in \KK_{\va}^t}\ip{\mP^{\perp}\vg}{\va}$ using Gaussian concentration. Since $\KK_{\va}^t\subset t\B_2^n$, the function $\vx \mapsto \sup_{\va \in\KK_{\va}^t}\ip{\mP^{\perp}\vx}{\va}$ has Lipschitz norm at most $t$. Indeed,
\begin{align*}
  \sup_{\va \in\KK_{\va}^t}\ip{\mP^{\perp}\vx}{\va} - \sup_{\va \in\KK_{\va}^t}\ip{\mP^{\perp}\vy}{\va} & \leq \ip{\mP^{\perp}\vx}{\tilde{\va}} - \ip{\mP^{\perp}\vy}{\tilde{\va}} \\
     & \leq  \|\tilde{\va}\|_2 \cdot \|\vx - \vy\|_2\\
     & \leq  t \cdot  \|\vx - \vy\|_2,
\end{align*}
where we choose $\tilde{\va}$ such that $\sup_{\va \in\KK_{\va}^t}\ip{\mP^{\perp}\vx}{\va} =  \ip{\mP^{\perp}\vx}{\tilde{\va}} $.

Therefore, Gaussian concentration inequality (Fact \ref{gaussian concentration}) implies that
	\begin{align*}
	\sup_{\va \in \KK_{\va}^t}\ip{\mP^{\perp}\vg}{\va} & \leq \E\sup_{\va \in \KK_{\va}^t}\ip{\mP^{\perp}\vg}{\va} +t\epsilon\sqrt{m}\\
                                                       &\leq \omega(\KK_{\va}^t)+t\epsilon\sqrt{m}
	\end{align*}
	holds with probability at least $1 - \exp(-c\epsilon^2m)$. The second inequality holds because
	\begin{align*}
	\E \sup_{\va \in \KK_{\va}^t}\ip{\mP^{\perp}\vg}{\va}& =\E \sup_{\va \in \KK_{\va}^t}\ip{\mP^{\perp}{\vg}+\E \mP{\vg}}{\va}\\
	&\leq \E \sup_{\va \in \KK_{\va}^t}\ip{\mP^{\perp}{\vg}+\mP{\vg}}{\va}\\
	&= \E \sup_{\va \in \KK_{\va}^t}\ip{{\vg}}{\va}=\omega(\KK_{\va}^t),
	\end{align*}
	where the inequality follows from the independence of $\mP{\vg}$ and $\mP^{\perp}{\vg}$ and  Jensen's inequality.

	Taking union bound yields, with probability at least $1 - 2 \exp \Big( -\frac{c m \sigma^4}{K^4} \Big)-\exp(-c\epsilon^2m)$,
	\begin{align*}
	E_2\leq 2\sqrt{m}\sigma\big[\omega(\KK_{\va}^t)+t\epsilon\sqrt{m}\big].
	\end{align*}
	
	Putting everything together, we conclude that, for any $0<s \leq \sqrt{m}$ (noting that $\epsilon=s/\sqrt{m}$),
	\begin{align*}
	\sup_{\va \in \KK_{\va}^t} \ip{\mPhi\va}{\vz}
	&\leq E_1+E_2\\
	&\leq m\epsilon\sigma t + 2\sqrt{m}\sigma\big[\omega(\KK_{\va}^t)+t\epsilon\sqrt{m}\big]\\
	&=\sqrt{m}\big[2\omega(\KK_{\va}^t)\sigma+2st\sigma +st\sigma\big]\\
	&\leq 3\sqrt{m}\big[\omega(\KK_{\va}^t)\sigma+st\sigma\big]
	\end{align*}
	holds with probability at least
	\begin{align*}
	&1-2 \exp \Big( -\frac{c m \epsilon^2 \sigma^2}{K^2} \Big)-2 \exp \Big( -\frac{c m \sigma^4}{K^4} \Big)-\exp(-c\epsilon^2m)\\
	&\ge 1-2 \exp \Big( -\frac{c'm\epsilon^2 \sigma^4}{K^4} \Big)=1-2 \exp \Big( -\frac{c's^2 \sigma^4}{K^4} \Big).
	\end{align*}
	Here we have used again that $\sigma^2\leq CK^2$ and $\epsilon \leq 1$. Thus we complete the proof.

\bibliographystyle{IEEEtran}
\bibliography{IEEEabrv,myref}

\end{document}